\documentclass[12pt,english]{article}
\usepackage{amssymb}
\usepackage{amsfonts}
\usepackage{amsmath}
\usepackage{color}
\usepackage{geometry}
\usepackage[authoryear]{natbib}

\usepackage[colorlinks=true,breaklinks=true,bookmarks=true,urlcolor=blue,
     citecolor=blue,linkcolor=blue,bookmarksopen=false,draft=false]{hyperref}

\makeindex
\newtheorem{theorem}{Theorem}

\newtheorem{definition}{Definition}

\newtheorem{proposition}{Proposition}
\newtheorem{remark}{Remark}

\newenvironment{proof}[1][Proof]{\noindent\textbf{#1.} }{\ \rule{0.5em}{0.5em}}

\begin{document}

\title{An axiomatic derivation of Condorcet-consistent social decision rules}
\author{Aurelien YONTA MEKUKO, Issofa MOYOUWOU, \and Matias N\'{U}\~{N}EZ,
Nicolas Gabriel ANDJIGA \and \thanks{%
Aurelien YONTA MEKUKO (\textbf{corresponding author}), Higher Teacher
Training College of Yaounde, P.O. 47, Yaounde, Cameroon, University of
Yaounde I. Email: aurelienmekuko@ens.cm}\thanks{%
Issofa MOYOUWOU, Higher Teacher Trainning College of Yaounde, P.O. 47
Yaounde, Cameroon, University of Yaounde I. Email: issofa.moyouwou@ens.cm}%
\thanks{%
Matias N\'{U}\~{N}EZ, Laboratoire LAMSADE, Universit\'{e} Paris Dauphine,
75775 Paris cedex 16 France. Email: matias.numez@dauphine.fr}\thanks{%
Nicolas Gabriel ANDJIGA, Higher Teacher Training College of Yaounde, P.O. 47
Yaounde, Cameroon, Unviversity of Yaounde I. Email: andjiga2002@yahoo.fr}}
\date{}
\maketitle

\begin{abstract}
A social decision rule (SDR) is any non empty set-valued map that associates
any profile of individual preferences with the set of (winning)
alternatives. An SDR is Condorcet-consistent if it selects the set of
Condorcet winners whenever this later is non empty. We propose a
characterization of Condorcet consistent SDRs with a set of minimal axioms.
It appears that all these rules satisfy a weaker Condorcet principle - the
\textit{top consistency} - which is not explicitly based on majority
comparisons while all scoring rules fail to meet it. We also propose an
alternative characterization of this class of rules using Maskin
monotonicity.
\end{abstract}

 Social decision rule - Condorcet-consistency - Top consistency - Maskin monotonicity

\section{Introduction}

The axiomatic literature on social decision rules (SDRs) has emphasized on
two main families that stand out due to their practical and appealing
properties: scoring SDRs and Condorcet-consistent SDRs. On the one hand,
scoring SDRs are the rules in which each voter submits a ballot that assigns
some number of points to each alternative, and the winners are the
alternatives with the maximum total number of points. As set-valued
functions, plurality rule, the Borda rule and approval voting rule belong to
this family. These rules are widely investigated and several axiomatizations
(with different degrees of generality) have been provided; seen \cite%
{Young1975} and \cite{Smith73}; or \cite{Myerson1995}, \cite%
{van2006characteristic}, \cite{Pivato2013} and \cite{andjiga2014metric} for
further analysis.

On the other hand, no axiomatization result covers the whole class of
Condorcet-consistent SDRs on which we focus in this paper. Under a
Condorcet-consistent SDR, voters submit each a ranking of the alternatives.
The outcome is then the set of Condorcet winners (CW)\footnote{%
Given a profile of individual preferences, a Condorcet winner (CW) is any
alternative who wins or ties all pairwise majority vote comparisons with any
other alternatives.} whenever the later is non empty. It can be undoubtedly
argued that the idea of CW is a central concept in voting theory \footnote{%
Note that some works on strategic voting theory have underlined the relation
between the equilibrium winners under Approval voting and the selection of
the Condorcet Winner (see \cite{Las2} and \cite{CourtinNunez2014} among
others); see also \cite{CrepelRieucau2005} for historical aspects or \cite%
{Gehrlein2006} for a comprehensive study of the probability that a CW exists
as well as the ability of various voting rules to fit the Condorcet
principle.}. Numerous studies have demonstrated the existence of interesting
and intuitive Condorcet-consistent SDRs\footnote{%
The literature on these rules is vast. See for a few examples \cite{Copeland}%
, \cite{Slater}, \cite{Schwartz1972}, \cite{Fishburn}, \cite{Dutta1988},
\cite{Schwartz1990} or \cite{Laffond1993} among others.}. Clearly two
distinct Condorcet-consistent SDRs differ just on the profiles which admit
no CW.

Within the framework of metric rationalizability\footnote{%
Distance rationalizability of SDRs requires that selected alternatives
should be the most preferred alternatives in the closest consensus profile,
closest been measure with a metric or a monometric (for instance, see \cite%
{perez2017monometrics}) .}, \cite{Elkind2012} nicely characterize several
Condorcet-consistent SDRs amoung which the Young rule and the Maximin rule.
\cite{Henriet1985} provides an axiomatic characterization for the Copeland
choice rule\footnote{%
A choice rule $f$ is defined as a mapping which associates to each set of
alternatives $A$ and each binary relation $R$ on $A$ a choice function $%
f(A,R,.)$. The choice function $f(A,R,.)$ associates to each nonempty subset
$B$ of $A$ the nonempty subset $f(A,R,B)$ of $B$, the set of winners when
the set of competing alternatives is $B$.}. Our objective is to provide some
characteristic features of the class of all Condorcet-consistent SDRs in
terms of a set of minimal axioms. By so doing, our approach contrasts with
previous normative works. To achieve this, we bear our attention on\ $2$%
-profiles which are profiles such that there are two alternatives
unanimously ranked above others. We present some axioms of coherence for
these profiles. Some of our axioms are simply restricted versions of some
usual axioms used to characterize the majority rule with two alternatives
(see \cite{may1952set}) or more (see \cite{CampbellKElly2000}, \cite%
{AsanSanver2002} or \cite{Woeginger} among others). For instance, we
introduce the \emph{top anonimity} (TA) and the \emph{top neutrality} (TN)
axioms respectively as the mild requirements that no permutation of voters
in a $2$-profile affects the winning set and any permutation of alternatives
in a $2$-profile emerges to permuting the winning set accordingly.
Similarly, to state our \emph{top monotonicity} (TM) axiom, consider three
alternatives $x,y$ and $z$; and two profiles $R$ and $Q$ such that $Q$ is
obtained from $R$ when some voter moves $x$ above $z$ while no voter moves
down $x$ nowhere. Then (TM) states that, if $x$ is selected when $x$ and $y$
are moved to the top of each voter's preference in $R$, then when $x$ and $y$
are moved to the top of each voter's preference in $Q$, $x$ is still
selected but not $z$. Roughly, TM guarantees that from a profile to a $2$%
-profile, an improvement of the ranking of an alternative is never harmful;
and the deterioration of the ranking of an alternative is.

We introduce some new axioms. \emph{top rationality} (TR) axiom states that
in a $2$-profile, at least one of the two unanimously top-ranked
alternatives should belong to the winning set. The two other newly
introduced axioms - \emph{weak top consistency}\textit{\ (}WTC\textit{)} and
\emph{top consistency}\textit{\ }(TC) are weaker versions of the Condorcet
principle. The $WTC$ axiom can be stated as follows: given any profile $R$
and any pair $\left\{ x,y\right\} $ of alternatives, the top-shift profile $%
R^{\{x,y\}}$ is the $2$-profile obtained by \textquotedblleft\ moving"
alternatives $x$ and $y$ at the top of voter preferences without any change
in their relative rankings. Then Weak top consistency requires that for a
given profile, whenever there exists an alternative that is selected each
time it is top-shifted with any other alternative, then this alternative is
selected. This condition can be viewed as some sort of Independence of Least
Preferred Alternatives. Indeed, if, given a profile $R$, whenever we
top-shift one alternative $x$ with any other alternative $y$, $x$ is the
winning set, it seems intuitive that $x$ should be in the winning set of $R$%
. The $TC$ condition requires that the winning set of the SDR consists of
all alternatives that are always selected each time they are top-shifted
with any other alternative.

It turns out that top consistency is the new frontier of
Condorcet-consistent SDRs that excludes all scoring SDRs. It is shown that
all Condorcet-consistent SDRs satisfy the top consistency axiom while all
scoring SDRs fail to meet it. Our characterization hence states that an SDR
is a Condorcet-consistent rule if and only it satisfies the previously
described axioms $TA$, $TN$, $TM$, $TR$ and $TC$. Futhermore, we prove that
the set of axioms is minimal, in the sense that, by omitting any single
axiom, there exists an SDR that satisfies all the other axioms but is not
Condorcet-consitent.

Finally and in order to shed some light on the role of the top consistency
axiom, we focus on the profiles which always admit a Condorcet Winner. In
this restricted domain, we prove that Condorcet-consistent SDRs satisfy the
wellknown Maskin monotonicity ($MM$) axiom and $WTC$; while in the
unrestricted domain, $MM$ fails to be satisfied. The condition of $MM$ is
known to be quite demanding as illustrated by the literature in Nash
implementation (see \cite{maskin1999nash} or \cite{mckelvey1989game}). In
this restricted domain, we prove that $WTC$ and $MM$ is equivalent to $TC$,
underlining the logic behind the top consistency condition. This allows us
to derive another axiomatic characterization of Condorcet-consistent rules
involving Maskin monotonicity.

The paper is organized as follows. In Section 2, we introduce basic
notations and definitions and formally describe our axioms. Some particular
highlights on those axioms are provided in Section 3 followed by the main
result which is a characterization of Condorcet-consistent SDRs. It is also
shown that our axioms are minimal. Section 4 is devoted to an alternative
characterization on Condorcet domain with the help of Maskin monotonicity.

\section{Notations and definitions}

Let $N=\left\{ 1,2,...,n\right\} $ denote a finite set of $n$ voters with $%
n\geq 2$ and $A$ a finite set of $m$ alternatives with $m\geq 3$. Voter
preference relations are defined over $A$ and are assumed to be weak orders
(complete and transitive binary relations on $A$). The set of all weak
orders on $A$ is denoted $W$. A (preference) profile is an $n-$tuple $%
R=(R_{1},R_{2},...,R_{n})$ of weak orders where the $i^{th}$ component $%
R_{i} $ of $R$ stands for voter $i$'s preference relation. The set of all
possible profiles is denoted $W^{N}$. Given $R\in W^{N}$ and $i\in N$:

\begin{itemize}
\item for any nonempty subset $B$ of $A$, $R_{i}|_{B}$ is the restriction of
$R_{i}$ on $B$;

\item for any partition $\left\{ A_{1},A_{2}\right\} $ of $A$, we write $%
Q_{i}=R_{i}|_{A_{1}}R_{i}|_{A_{2}}$ if voter $i$ strictly prefers each
alternative in $A_{1}$ to each alternative in $A_{2}$, alternatives in $A_{1}
$ are ranked according to $R_{i}|_{A_{1}}$ while alternatives in $A_{2}$ are
ranked according to $R_{i}|_{A_{2}}$;

\item $\succ _{R_{i}}$ and $\sim _{R_{i}}$ are respectively the asymmetric
component and the symmetric component of $R_{i}$;

\item For any pair $\{x,y\}\subseteq A$,

\begin{itemize}
\item $n\left( x,y,R\right) =\#\left\{ i\in N:x\succ _{R_{i}}y\right\} $. In
other words, $n\left( x,y,R\right) $ stands for the number of voters who
strictly prefer $x$ to $y$ in the profile $R$.

\item $x\succcurlyeq _{R_{i}}y$ holds if $x\succ _{R_{i}}y$ or $x\sim
_{R_{i}}y$;

\item $R_{i}^{\left\{ x,y\right\} }=R_{i}|_{\left\{ x,y\right\}
}R_{i}|_{A\backslash \left\{ x,y\right\} }$ stands for the weak order
obtained from $R_{i}$ by only moving to the top $x$ and $y$ without changing
their relative ranking; and $R^{\left\{ x,y\right\} }$ is the $2$-profile
obtained from $R$ by substituting $R_{i}^{\left\{ x,y\right\} }$ to $R_{i}$
for each $i\in N$; $R^{\left\{ x,y\right\} }$ is also called the top-shift
profile of $x$ and $y$ from $R$.

\item We simply write $R_{i}|_{\left\{ x,y\right\} }=xy$ if $x\succ
_{R_{i}}y $ and $R_{i}|_{\left\{ x,y\right\} }=\left( xy\right) $ if $x\sim
_{R_{i}}y$. For example, $xyR_{i}|_{A\backslash \left\{ x,y\right\} }$
stands for the weak order in which $x$ is first, $y$ is second and
alternatives other than $x$ and $y$ are ranked lower than $y$ and according
to $R_{i}$.
\end{itemize}
\end{itemize}

A social decision rule (SDR) is a mapping $C$ from $W^{N}$ to $%
2^{A}\backslash \{\emptyset \}$, the set of nonempty subsets of $A$. We now
introduce two known classes of SDRs: Condorcet-consistent SDRs and $L$%
-scoring SDRs.

\begin{definition}
\begin{enumerate}
\item For any $R\in W^{N}$ and any pair $\{x,y\}\subseteq A$, we say that $x$
beats $y$ in a pairwise majority vote, denoted $xMy$, if $n\left(
x,y,R\right) >n\left( y,x,R\right) $.

Moreover, $x$ is a Condorcet Winner if $n\left( x,y,R\right) \geq n\left(
y,x,R\right) ,\forall y\neq x.$

The set of all Condorcet winners (possibly empty) in $R$ is denoted by $%
CW\left( R\right) $.

\item An SDR $C$ is Condorcet-consistent if for all $R\in W^{N},$ $C\left(
R\right) =CW\left( R\right) $ whenever $CW\left( R\right) \neq \emptyset .$
\end{enumerate}
\end{definition}

Denote by $L$ the set of all linear orders (or strict orders) on $A$ and by $%
L^{N\text{ }}$ the set of all profiles of linear orders. The rank of an
alternative $x$ with respect to a given linear order $l$\ denoted by $%
rg\left( x,l\right) $ is the total number of alternatives $y$ such that $%
y\succcurlyeq _{l}x$ and a scoring vector is any $m$-tuple $\alpha =\left(
\alpha _{1},\alpha _{2},...,\alpha _{m}\right) $ of real numbers such that $%
\alpha _{1}\geq \alpha _{2}\geq ...\geq \alpha _{m}$ with $\alpha
_{1}>\alpha _{m}$. Given $R\in L^{N}$, a scoring vector $\alpha $ and an
alternative $x$, we define the score of $x$ in $R$ as $S_{\alpha }\left(
x,R\right) =$ $\sum\limits_{i\in N}\alpha _{rg\left( x,R_{i}\right) }$. We
denote by $C_{\alpha }\left( R\right) $ the subset of $A$ defined as follow:
\begin{equation*}
C_{\alpha }\left( R\right) =\left\{ x\in A:S_{\alpha }\left( x,R\right) \geq
S_{\alpha }\left( y,R\right) ,\forall y\neq x\right\} .
\end{equation*}

\begin{definition}
An SDR is an $L$-scoring SDR if there exists a scoring vector $\alpha $ such
that $C\left( R\right) =C_{\alpha }\left( R\right) $ for all $R\in L^{N}$.
\end{definition}

Note that for a scoring vector $\alpha $, the mapping $C_{\alpha }$, that
associates each profile $R$ of linear orders with the subset $C_{\alpha
}\left( R\right) $ of $A$, is a scoring SDR on $L^{N}$. Therefore any $L$%
-scoring SDR can be viewed as an extension of a scoring SDR from $L^{N}$ to $%
W^{N}$.

\begin{definition}
Given an SDR $C$, for any preference profile $R\in W^{N}$, the nice set of $%
R $, denoted by $\mathcal{N}_{C}(R)$, is defined as follows:
\begin{equation*}
\mathcal{N}_{C}(R)=\{x\in A:x\in C(R^{\left\{ x,y\right\} }),\forall y\in
A\setminus \{x\}\}.
\end{equation*}
\end{definition}

Given an SDR, the nice set of a given profile is the collection of all
alternatives that are always winning each time they are top-shifted with any
other alternative.

\begin{definition}
An SDR $C$ satisfies weak top consistency (WTC) if for any $R\in W^{N}$, $%
C(R)\supseteq \mathcal{N}_{C}(R)$.
\end{definition}

According to weak top consistency, any alternative that belongs to the nice
set for a given profile is selected by a given SDR.

\begin{definition}
An SDR $C$ satisfies top consistency (TC) if for any $R\in W^{N}$ with $%
\mathcal{N}_{C}(R)\neq \emptyset $, $C(R)=\mathcal{N}_{C}(R)$.
\end{definition}

Top consistency requires that given an SDR, the winning set is exactly the
nice set whenever it is nonempty. It is obvious that each SDR that satisfies
TC also satisfies WTC.

\medskip To introduce the next two definitions, we need further notations.
We denote by $S_{N}$ (respectively $S_{A}$) the set of all permutations of $%
N $ (respectively $A$). Given $R\in W^{N}$, $i\in N$, $\pi \in S_{N}$ and $%
\sigma \in S_{A}$: (i) $R_{\pi }=\left( R_{\pi \left( 1\right) },R_{\pi
\left( 2\right) },...,R_{\pi \left( n\right) }\right) $ is the profile
obtained from $R$ by permuting voter preference relations with respect to $%
\pi $ in such a way that voter $i$ is now affected voter $j$'s preference
relation with $j=\pi \left( i\right) $; (ii) $\sigma \left( R\right) =\left(
\sigma \left( R^{1}\right) ,\sigma \left( R^{2}\right) ,...,\sigma \left(
R^{n}\right) \right) $ is the profile obtained from $R$ after relabeling
alternatives according to $\sigma $, that is for all $a,b\in A$ and for all $%
i\in N,$ $a\succ _{R_{i}}b\Longleftrightarrow \sigma \left( a\right) \succ
_{\sigma \left( R_{i}\right) }\sigma \left( b\right) $; (iii) given a non
empty subset $B$ of $A$, $\sigma \left( B\right) =\left\{ \sigma \left(
b\right) :b\in B\right\} $.\medskip

A $2$-profile is a profile in which there exist two alternatives ranked
above any other alternatives. Let $W_{2}^{N}$ denote the set of all $2$%
-profiles. That is:
\begin{equation*}
R\in W_{2}^{N}\Longleftrightarrow \exists \left\{ x,y\right\} \subseteq
A:\forall z\in A\backslash \left\{ x,y\right\} ,\forall i\in N,x\succ
_{R_{i}}z\text{ and }y\succ _{R_{i}}z\text{.}
\end{equation*}

\begin{definition}
Given an SDR $C$,

\begin{enumerate}
\item $C$ satisfies top neutrality (TN) if $\forall R\in W_{2}^{N},\forall
\sigma \in S_{A},C\left( \sigma \left( R\right) \right) =\sigma \left(
C\left( R\right) \right) .$

\item $C$ satisfies top anonymity (TA) if $\forall R\in W_{2}^{N},\forall
\pi \in S_{N}:C\left( R_{\pi }\right) =C\left( R\right) .$

\item $C$ is top symmetric (TS) if $C$ is both TN and TA.
\end{enumerate}
\end{definition}

Note that top neutrality and top anonymity are respectively the restrictions
of the well-known neutrality axiom and anonymity axiom from $W^{N}$ to $%
W_{2}^{N}$. Top symmetric then amounts to saying that both names of
candidates and names of voters should not play any role in determining
winning alternatives over $W_{2}^{N}$.\medskip

Monotonicity properties are interprofile criteria stipulating that from a
profile to another, additional support is never harmful for an alternative.
To state the next axiom that can be viewed as a monotonicity property
between profiles in $W_{2}^{N}$, we use the following notations to precise
what should be considered as an additional support. Given $R,Q\in W^{N}$, we
write $R\vartriangleright ^{x,y}Q$ if (i) $\forall i\in N,$ $\forall z\in A,$
$x\succ _{R_{i}}z\Rightarrow x\succ _{Q_{i}}z$\ and $x\sim
_{R_{i}}z\Rightarrow x\succcurlyeq _{Q_{i}}z$, (the rank of $x$ in voter
preferences never decreases from $R$ to $Q$); and (ii) $y\succcurlyeq
_{R_{i}}x$ and $x\succ _{Q_{i}}y$ for some $i\in N$. When $%
R\vartriangleright ^{x,y}Q$ holds, we say that $Q$ is an additional support
of $x$ against $y$ from $R$ to $Q$.

\begin{definition}
An SDR $C$ satisfies top monotonicity (TM) if $\forall R,Q\in W^{N},\forall
\left\{ x,y\right\} \subseteq A,\forall z\in A\backslash \left\{ x\right\} :$%
\begin{equation*}
\left( x\in C\left( R^{\left\{ x,y\right\} }\right) \text{ and }%
R\vartriangleright ^{x,z}Q\right) \Rightarrow x\in C\left( Q^{\left\{
x,y\right\} }\right) \text{ and }z\not\in C\left( Q^{\left\{ x,y\right\}
}\right) .
\end{equation*}
\end{definition}

Assume that $x$ is selected in a profile $R$ when top-shifted with another
alternative $y$. Then $TM$ requires that any additional support of $x$
against an alternative $z$ (possibly $y$) from $R$ to a new profile $Q$
results in $Q^{\left\{ x,y\right\} }$ in dismissing $z$ from the winning set
meanwhile $x$ is still winning.

\begin{definition}
An SDR $C$ satisfies top rationality (TR) if

$\forall R\in W^{N},\forall \left\{ x,y\right\} \subseteq A:C\left(
R^{\left\{ x,y\right\} }\right) \cap \left\{ x,y\right\} \neq \emptyset $.
\end{definition}

Top rationality is a very weak requirement: whenever two alternatives are
top-shifted in a profile, at least one of them is selected.

\section{Highlights on axioms and characterization}

We provide here a complete characterization of Condorcet-consistent SDRs.
But before, we present some results which highlight some properties of the
axioms we use.

\subsection{Top consistency}

In this section, we show that TC constitutes a border line between
Condorcet-consistent SDRs and $L$-scoring SDRs. More precisely, it is shown
that all Condorcet-consistent SDRs satisfy TC while all $L$-scoring SDRs
fail to meet it.

\begin{proposition}
Any Condorcet-consistent SDR satisfies TC.
\end{proposition}

\begin{proof}
Assume that $C$ is a Condorcet-consistent SDR. Consider $R\in W^{N}$ such
that $\mathcal{N}_{C}\left( R\right) \neq \emptyset $. We prove that $%
C\left( R\right) =\mathcal{N}_{C}\left( R\right) $.\medskip

On the one hand, consider $y\in \mathcal{N}_{C}\left( R\right) $ and let $%
z\in A\backslash \left\{ y\right\} $. Note that, with respect to $R^{\left\{
z,y\right\} }$, both $z$ and $y$ beat any other alternative $t\in
A\backslash \left\{ z,y\right\} $ in a pairwise majority duel. Suppose that $%
y$ is beaten by $z$ in $R^{\left\{ z,y\right\} }$. Then $z$ is the unique
Condorcet winner in $R^{\left\{ z,y\right\} }$; that is $CW\left( R^{\left\{
z,y\right\} }\right) =\left\{ z\right\} $. Since $C$ is
Condorcet-consistent, $C\left( R^{\left\{ z,y\right\} }\right) =\left\{
z\right\} $ and $y\notin C\left( R^{\left\{ z,y\right\} }\right) $. A
contradiction arises since $y\in \mathcal{N}_{C}\left( R\right) $.
Therefore, $y$ is not beaten by $z$ in $R^{\left\{ z,y\right\} }$. By
definition of $R^{\left\{ z,y\right\} }$, $y$ is not beaten by $z$ in $R$.
This is true for all $z\in A\backslash \left\{ y\right\} $. It follows that $%
y\in CW\left( R\right) \neq \emptyset $. Since $C$ is Condorcet-consistent
and $CW\left( R\right) \neq \emptyset $, then $C\left( R\right) =CW\left(
R\right) $. Hence $y\in C\left( R\right) $. This proves that $\mathcal{N}%
_{C}\left( R\right) \subseteq C\left( R\right) $.\medskip

On the other hand, consider $y\in C\left( R\right) $ and let $z\in
A\backslash \left\{ y\right\} $. By assumption, $\mathcal{N}_{C}\left(
R\right) \neq \emptyset $. Choose an alternative $x\in \mathcal{N}_{C}\left(
R\right) $. As we just show, $x\in CW\left( R\right) $. Therefore $CW\left(
R\right) \neq \emptyset $ and $C\left( R\right) =CW\left( R\right) $. This
implies that $y\in CW\left( R\right) $. Moreover, $y\in CW\left( R^{\left\{
z,y\right\} }\right) \neq \emptyset $ by definition of $R^{\left\{
z,y\right\} }$. Since $C$ is Condorcet-consistent and $CW\left( R^{\left\{
z,y\right\} }\right) \neq \emptyset $, it follows that $CW\left( R^{\left\{
z,y\right\} }\right) =C\left( R^{\left\{ z,y\right\} }\right) $. Hence $y\in
C\left( R^{\left\{ z,y\right\} }\right) $. This proves that $y\in C\left(
R^{\left\{ z,y\right\} }\right) $ for all $z\in A\backslash \left\{
y\right\} $. Thus $y\in \mathcal{N}_{C}\left( R\right) $. We conclude that $%
C\left( R\right) \subseteq \mathcal{N}_{C}\left( R\right) $.
\end{proof}

\begin{proposition}
Assume that $m\geq 3$ and $n\geq 2$. If $n\neq 3$ then any $L$-scoring SDR
fails to satisfies TC.
\end{proposition}

\begin{proof}
Assume that $m\geq 3$ and $n\geq 2$. Let $C$ be a $L$-scoring SDR. Then by
definition, there exists a scoring vector $\alpha $ such that $C\left(
R\right) =C_{\alpha }\left( R\right) $ for all $R\in L^{N}$. In what
follows: (i) $a$, $b$ and $c$ are distinct alternatives and $B=A\backslash
\left\{ a,b,c\right\} $; (ii) $l$ is a given linear order on $B$; and (iii) $%
xyz\left[ l\right] $ (respectively $xy\left[ l\right] z$) corresponds to the
linear order in which $x$ is ranked first, $y$ is second, $z$ is third
(respectively bottom ranked) and alternatives in $B$ are ranked according to
$l$ with $\left\{ x,y,z\right\} =\left\{ a,b,c\right\} $.

Case 1 : $n$ is even and $n\geq 2$. We pose $n=2p$ and $N=N_{1}\cup N_{2}$
with $\left\vert N_{1}\right\vert =\left\vert N_{2}\right\vert =p$.

First assume that $\alpha _{1}=\alpha _{2}$ or $\alpha _{1}>\alpha
_{2}>\alpha _{3}$. Let $R$ be the profile such that for each $i\in N$, $%
R_{i}=ab\left[ l\right] c$ if $i\in N_{1}$, $R_{i}=cab\left[ l\right] $ if $%
i\in N_{2}$. In both cases, note that $S_{\alpha }\left( a,R\right)
-S_{\alpha }\left( c,R\right) =p\left( \alpha _{2}-\alpha _{m}\right) >0$.
It follows that $c\notin C_{\alpha }\left( R\right) =C\left( R\right) $. But
for all $x\in A\backslash \left\{ c\right\} $, $c\in C_{\alpha }\left(
R^{\left\{ c,x\right\} }\right) =C\left( R^{\left\{ c,x\right\} }\right) $.
Therefore $c\in \mathcal{N}_{C}\left( R\right) $. Clearly $C\left( R\right)
\neq \mathcal{N}_{C}\left( R\right) \neq \emptyset $. Thus $C$ does not
satisfy TC.

Now assume that $\alpha _{1}>\alpha _{2}=\alpha _{3}$. Let $R$ be the
profile such that for each $i\in N$, $R_{i}=ab\left[ l\right] c$ if $i\in
N_{1}$, $R_{i}=cba\left[ l\right] $ if $i\in N_{2}$. We have $S_{\alpha
}\left( a,R\right) -S_{\alpha }\left( b,R\right) =p\left( \alpha _{1}-\alpha
_{2}\right) >0$. Therefore $b\notin C_{\alpha }\left( R\right) =C\left(
R\right) $. But for all $x\in A\backslash \left\{ b\right\} $, $b\in
C_{\alpha }\left( R^{\left\{ b,x\right\} }\right) =C\left( R^{\left\{
b,x\right\} }\right) $. This implies that $b\in \mathcal{N}_{C}\left(
P\right) $. Clearly $C\left( P\right) \neq \mathcal{N}_{C}\left( P\right)
\neq \emptyset $. This proves that $C$ does not satisfy TC.

Case 2 : $n$ is odd and $n\geq 5$. We write $n=3+2p$ and $N=\left\{
1,2,3\right\} \cup N_{1}\cup N_{2}$ with $\left\vert N_{1}\right\vert
=\left\vert N_{2}\right\vert =p\geq 1$. Consider the profile $R$ such that $%
R_{1}=abc\left[ l\right] $, $R_{2}=bca\left[ l\right] $, $R_{3}=cab\left[ l%
\right] $, $R_{i}=ab\left[ l\right] c$ if $i\in N_{1}$ and $R_{i}=ba\left[ l%
\right] c$ if $i\in N_{2}$.

First suppose that $\alpha _{1}=\alpha _{2}$. It can be checked that $%
\left\{ x,y\right\} \subseteq C_{\alpha }\left( R^{\left\{ x,y\right\}
}\right) $ for all $\left\{ x,y\right\} \subseteq A$. This implies that $%
\mathcal{N}_{C}\left( R\right) =A$. But $S_{\alpha }\left( a,R\right)
-S_{\alpha }\left( c,R\right) =p\left( \alpha _{1}+\alpha _{2}-2\alpha
_{m}\right) >0$. Hence $c\notin C_{\alpha }\left( R\right) =C\left( R\right)
$. Therefore $C\left( R\right) \neq \mathcal{N}_{C}\left( R\right) $ while $%
\mathcal{N}_{C}\left( R\right) \neq \emptyset $. Clearly, $C$ does not
satisfy TC.

Now suppose that $\alpha _{1}>\alpha _{2}$. Note that $a\in C_{\alpha
}\left( R^{\left\{ a,y\right\} }\right) $ for all $y\in A\backslash \left\{
a\right\} $. This implies that $a\in \mathcal{N}_{C}\left( R\right) \neq
\emptyset $. Moreover $S_{\alpha }\left( b,R\right) =S_{\alpha }\left(
a,R\right) \geq S_{\alpha }\left( x,R\right) $ for all $x\in A$. Thus $b\in
C_{\alpha }\left( R\right) =C\left( R\right) $. But $S_{\alpha }\left(
a,R^{\left\{ a,b\right\} }\right) -S_{\alpha }\left( b,R^{\left\{
a,b\right\} }\right) =\alpha _{1}-\alpha _{2}>0$. Therefore $b\notin
C_{\alpha }\left( R^{\left\{ a,b\right\} }\right) =C\left( R^{\left\{
a,b\right\} }\right) $. This implies that $b\notin C\left( R\right) $ and $%
C\left( R\right) \neq \mathcal{N}_{C}\left( R\right) \neq \emptyset $.
Clearly, $C$ does not satisfy TC.\medskip
\end{proof}

\subsection{Consequences of top monotonicity and top rationality}

The next results highlight some consequences of combining TM and TR.

\begin{proposition}
\label{prop TR TM to xy}Assume that $C$ satisfies TM and TR.

Then for all $R\in W^{N}$ and all $\left\{ x,y\right\} \subseteq A$, $%
C\left( R^{\left\{ x,y\right\} }\right) \subseteq \left\{ x,y\right\} $.
\end{proposition}

\begin{proof}
Assume that $C$ satisfies TM and TR. Consider $R\in W^{N}$ and $\left\{
x,y\right\} \subseteq A$. By TR, $C\left( R^{\left\{ x,y\right\} }\right)
\cap \left\{ x,y\right\} \neq \emptyset $. Without lost of generality,
assume that $x\in C\left( R^{\left\{ x,y\right\} }\right) $. Consider the
profile $Q$ define by $Q_{i}=R_{i}|_{A\backslash \left\{ x,y\right\}
}R_{i}|_{\left\{ x,y\right\} }$ for all $i\in N$. Note that $Q$ is obtained
from $R$ by only moving $x$ and $y$ to the bottom in each voter preference
without changing their relative ranking. Also note that $Q^{\left\{
x,y\right\} }=R^{\left\{ x,y\right\} }$ and that $Q\vartriangleright
^{x,z}R^{\left\{ x,y\right\} }$ for each $z\in A\backslash \left\{
x,y\right\} $. Since $x\in C\left( R^{\left\{ x,y\right\} }\right) $, then $%
x\in C\left( Q^{\left\{ x,y\right\} }\right) $ and by TM, $z\notin C\left(
R^{\left\{ x,y\right\} }\right) $ for any $z\in A\backslash \left\{
x,y\right\} $. Therefore $C\left( R^{\left\{ x,y\right\} }\right) \subseteq
\left\{ x,y\right\} $.
\end{proof}

\begin{proposition}
\label{prop TR TM to =12}Assume that $C$ satisfies TM and TR.

For all $R,Q\in W^{N}$ and all $\left\{ x,y\right\} \subseteq A$, if $%
R_{i}|_{\left\{ x,y\right\} }=Q_{i}|_{\left\{ x,y\right\} }$ for all $i\in N$%
, then $C\left( R^{\left\{ x,y\right\} }\right) =C\left( Q^{\left\{
x,y\right\} }\right) $.
\end{proposition}

\begin{proof}
Assume that $C$ satisfies TM and TR. Consider $R,Q\in W^{N}$ and $\left\{
x,y\right\} \subseteq A$ such that $R_{i}|_{\left\{ x,y\right\}
}=Q_{i}|_{\left\{ x,y\right\} }$ for all $i\in N$. Since $C$ satisfies TM
and TR, then by Proposition \ref{prop TR TM to xy}, $C\left( R^{\left\{
x,y\right\} }\right) \subseteq \left\{ x,y\right\} $. Suppose that $x\in
C\left( R^{\left\{ x,y\right\} }\right) $, then consider the profile $H$
define by $H_{i}=R_{i}|_{A\backslash \left\{ x,y\right\} }R_{i}|_{\left\{
x,y\right\} }$ for all $i\in N$. Note that $H^{\left\{ x,y\right\}
}=R^{\left\{ x,y\right\} }$ and that $H\vartriangleright ^{x,z}Q^{\left\{
x,y\right\} }$ for any $z\in A\backslash \left\{ x,y\right\} $. Since $x\in
C\left( R^{\left\{ x,y\right\} }\right) $, then $x\in C\left( H^{\left\{
x,y\right\} }\right) $ and by TM, $x\in C\left( Q^{\left\{ x,y\right\}
}\right) $. Therefore $C\left( R^{\left\{ x,y\right\} }\right) \subseteq
C\left( Q^{\left\{ x,y\right\} }\right) $. Similarly, we prove that $C\left(
Q^{\left\{ x,y\right\} }\right) \subseteq C\left( R^{\left\{ x,y\right\}
}\right) $. Hence $C\left( R^{\left\{ x,y\right\} }\right) =C\left(
Q^{\left\{ x,y\right\} }\right) $.
\end{proof}

\subsection{Unrestricted domain: a characterization}

\begin{proposition}
\label{prop CW included}Assume that $C:W^{N}\rightrightarrows A$ satisfies
WTC, TM, TS and TR.

Then if $x$ is a Condorcet winner in $R$, then $x\in C\left( R\right) $.
\end{proposition}

\begin{proof}
Suppose that $x$ is a Condorcet winner in $R$. Assume that $x\notin C\left(
R\right) $. By WTC, there exists $y\in A\backslash \left\{ x\right\} $ such
that $x\notin C\left( R^{\left\{ x,y\right\} }\right) $. Let $B=A\backslash
\left\{ x,y\right\} $.\medskip

Since $C$ satisfies TM and TR, then by Proposition \ref{prop TR TM to xy}, $%
C\left( R^{\left\{ x,y\right\} }\right) \subseteq \left\{ x,y\right\} $.
Since $C\left( R^{\left\{ x,y\right\} }\right) $ is nonempty, then $C\left(
R^{\left\{ x,y\right\} }\right) =\left\{ y\right\} $. Let $N_{1}=\left\{
i\in N,x\succ _{R_{i}}y\right\} ,$ $N_{2}=\left\{ i\in N,y\succ
_{R_{i}}x\right\} $ and $N_{3}=\left\{ i\in N,x\sim _{R_{i}}y\right\} $.
Since $x$ is a Condorcet winner in $R$, then $\left\vert N_{1}\right\vert
\geq \left\vert N_{2}\right\vert $. Thus $N_{1}=S_{1}\cup T_{1}$ with $%
\left\vert S_{1}\right\vert =\left\vert N_{2}\right\vert $ for some $%
T_{1}\subset N$. Moreover $R_{i}^{\left\{ x,y\right\} }=xyR_{i}|_{B}$ if $%
i\in N_{1}$, $R_{i}^{\left\{ x,y\right\} }=yxR_{i}|_{B}$ if $i\in N_{2}$ and
$R_{i}^{\left\{ x,y\right\} }=\left( xy\right) R_{i}|_{B}$ if $i\in N_{3}$.
Consider any one to one mapping $\nu $ from $S_{1}$ to $N_{2}$ and define a
permutation $\pi $ of $N$ as follows: $\pi \left( i\right) =\nu \left(
i\right) $ if $i\in S_{1}$, $\pi \left( i\right) =\nu ^{-1}\left( i\right) $
if $i\in N_{2}$ and $\pi \left( i\right) =i$ if $i\in T_{1}\cup N_{3}$. Note
that $\pi \left( S_{1}\right) =N_{2}$, $\pi \left( N_{2}\right) =S_{1}$, $%
\pi \left( T_{1}\cup N_{3}\right) =T_{1}\cup N_{3}$ and $\pi ^{-1}=\pi $.

\medskip

First consider the profile $Q$ defined by $Q_{i}=R_{i}|_{\left\{ x,y\right\}
}R_{\pi \left( i\right) }|_{B}$ for all $i\in N$. Note that $Q^{\left\{
x,y\right\} }=Q$ and that $R_{i}|_{\left\{ x,y\right\} }=Q_{i}|_{\left\{
x,y\right\} }$ for all $i\in N$. By Proposition \ref{prop TR TM to =12}, $%
C\left( Q^{\left\{ x,y\right\} }\right) =C\left( R^{\left\{ x,y\right\}
}\right) =$ $\left\{ y\right\} $. That is $C\left( Q\right) =\left\{
y\right\} $.\medskip

Now let $H$ be the profile obtained from $Q$ by only permuting $x$ and $y$.
That is $H=\sigma \left( Q\right) $ where $\sigma $ is the permutation of $A$
defined by $\sigma \left( x\right) =y$, $\sigma \left( y\right) =x$ and $%
\sigma \left( z\right) =z$ for all $z\in A\backslash \left\{ x,y\right\} $.
Then $H_{i}=yxR_{\pi \left( i\right) }|_{B}$ if $i\in N_{1}$, $%
H_{i}=xyR_{\pi \left( i\right) }|_{B}$ if $i\in N_{2}$ and $H_{i}=\left(
xy\right) R_{\pi \left( i\right) }|_{B}$ if $i\in N_{3}$. By TS
(particularly TN) , $C\left( H\right) =C\left( \sigma \left( Q\right)
\right) =\sigma \left( \left\{ y\right\} \right) =\left\{ x\right\} $. Also
consider the profile $U=H_{\pi }$. Since $\pi ^{-1}=\pi $, it follows that:
(i) for all $i\in S_{1}$, $\pi \left( i\right) \in N_{2}$ and $U_{i}=H_{\pi
\left( i\right) }=xyR_{\pi \left[ \pi \left( i\right) \right]
}|_{B}=xyR_{i}|_{B}=R_{i}$; (ii) for all $i\in T_{1}$, $\pi \left( i\right)
=i\in N_{1}$ and $U_{i}=H_{i}=yxR_{i}|_{B}$; (iii) for all $i\in N_{2}$, $%
\pi \left( i\right) \in S_{1}\subseteq N_{1}$ and $U_{i}=H_{\pi \left(
i\right) }=yxR_{\pi \left[ \pi \left( i\right) \right]
}|_{B}=yxR_{i}|_{B}=R_{i}$; and (iv) for all $i\in N_{3}$, $\pi \left(
i\right) =i$ and $U_{i}=H_{i}=\left( xy\right) R_{i}|_{B}=R_{i}$. By TS
(particularly TA), $C\left( U\right) =C\left( H\right) =\left\{ x\right\} $%
.\medskip

Finally, let $V$ be the profile obtained from $U$ by only reversing the
relative ranking of $x$ and $y$ for each player in $T_{1}$. Then $%
V_{i}=U_{i}=R_{i}$ for all $i\in S_{1}\cup N_{2}\cup N_{3}$ and $%
V_{i}=xyR_{i}|_{B}=R_{i}$ for $i\in T_{1}$. Hence $V=R^{\left\{ x,y\right\}
} $. Moreover $U\trianglerighteq ^{x,y}V=R$ and $x\in C\left( U\right)
=C\left( U^{\left\{ x,y\right\} }\right) $. Thus by TM, $x\in C\left(
V^{\left\{ x,y\right\} }\right) =C\left( R^{\left\{ x,y\right\} }\right) $.
That is a contradiction since $x\notin C\left( R^{\left\{ x,y\right\}
}\right) $.\medskip

In conclusion, $x\in C\left( R\right) $.
\end{proof}

\bigskip The following remark is important to ease the proof of the next
theorem.

\begin{remark}
\label{remark CW}Given a profile $R$ and two distinct alternatives $x$ and $%
y $ and for all profile $R$, by the definition of $R^{\left\{ x,y\right\} }$%
, only three possible cases may occur: $CW\left( R^{\left\{ x,y\right\}
}\right) =\left\{ x\right\} ,CW\left( R^{\left\{ x,y\right\} }\right)
=\left\{ y\right\} $ or $CW\left( R^{\left\{ x,y\right\} }\right) =\left\{
x,y\right\} $. Then $CW\left( R^{\left\{ x,y\right\} }\right) $ is always a
nonempty set. Therefore, if $C$ is a Condorcet-consistent SDR, then $C\left(
R^{\left\{ x,y\right\} }\right) =CW\left( R^{\left\{ x,y\right\} }\right) $.
\end{remark}

\begin{theorem}
\label{Theorem result}An SDR $C:W^{N}\rightrightarrows A$ is
Condorcet-consistent if and only if $C$ satisfies TC, TS, TM and TR.
\end{theorem}

\begin{proof}
Consider an SDR $C:W^{N}\rightrightarrows A$.

\textbf{Sufficiency:} assume that $C$ satisfies TC, TS, TM and TR. Consider
any profile in which $CW\left( R\right) \neq \emptyset $. By Proposition \ref%
{prop CW included}, $C\left( R\right) \supseteq CW\left( R\right) $. Then to
prove that $C\left( R\right) =CW\left( R\right) $, it is sufficient to show
that $C\left( R\right) \subseteq CW\left( R\right) $. Since $CW\left(
R\right) \neq \emptyset $, consider $a\in CW\left( R\right) $. For all $y\in
A\backslash \left\{ a\right\} $, $x\in CW\left( R^{\left\{ a,y\right\}
}\right) $ and by Proposition \ref{prop CW included}, $CW\left( R^{\left\{
a,y\right\} }\right) \subseteq C\left( R^{\left\{ a,y\right\} }\right) $.
Therefore for all $y\in A\backslash \left\{ a\right\} $, $x\in C\left(
R^{\left\{ a,y\right\} }\right) $. Hence $x\in \mathcal{N}_{C}(R)\neq
\emptyset $. Thus by TC, $C\left( R\right) =\mathcal{N}_{C}(R)=$ $z\in
A:\forall y\in A\backslash \{z\}$ $,z\in C(R^{\left\{ y,z\right\} })$%
.\medskip

Suppose that there exists $x\in C\left( R\right) $ such that $x\notin
CW\left( R\right) $.Therefore there exists $c\in A$ such that $\left\vert
N_{1}\right\vert <\left\vert N_{2}\right\vert $ where $N_{1}=\left\{ i\in
N,x\succ _{i}c\right\} $, $N_{2}=\left\{ i\in N,c\succ _{i}x\right\} $, $%
N_{3}=\left\{ i\in N,x\sim _{i}c\right\} $ and $N=N_{1}\cup N_{2}\cup N_{3}$%
. Note that $CW\left( R^{\left\{ c,x\right\} }\right) =\left\{ c\right\} $.
By Proposition \ref{prop CW included}, $c\in C\left( R^{\left\{ c,x\right\}
}\right) $. Since $x\in C\left( R\right) $, it follows that for all $y\in
A\backslash \left\{ x\right\} ,$ $x\in C\left( R^{\left\{ x,y\right\}
}\right) $. Hence $x\in C\left( R^{\left\{ c,x\right\} }\right) $. By
Proposition \ref{prop TR TM to xy}, $C\left( R^{\left\{ c,x\right\} }\right)
=\left\{ c,x\right\} $.\medskip

As in the proof of proposition \ref{prop CW included}, let $B=A\backslash
\left\{ c,x\right\} $ and $N_{2}=S_{2}\cup T_{2}$ such that $\left\vert
S_{2}\right\vert =\left\vert N_{1}\right\vert $ and $\left\vert
T_{2}\right\vert \geq 1$. Then $R_{i}^{\left\{ c,x\right\} }=xcR_{i}|_{B}$
if $i\in N_{1}$, $R_{i}^{\left\{ c,x\right\} }=cxR_{i}|_{B}$ if $i\in N_{2}$
and $R_{i}^{\left\{ c,x\right\} }=\left( xc\right) R_{i}|_{B}$ if $i\in
N_{3} $. Consider any one to one mapping $\nu $ from $N_{1}$ to $S_{2}$ and
define a permutation $\pi $ of $N$ as follows: $\pi \left( i\right) =\nu
\left( i\right) $ if $i\in N_{1}$, $\pi \left( i\right) =\nu ^{-1}\left(
i\right) $ if $i\in S_{2}$ and $\pi \left( i\right) =i$ if $i\in T_{2}\cup
N_{3}$. Clearly $\pi \left( N_{1}\right) =S_{2}$, $\pi \left( S_{2}\right)
=N_{1}$, $\pi \left( T_{2}\cup N_{3}\right) =T_{2}\cup N_{3}$ and $\pi
^{-1}=\pi $.\medskip

First consider the profile $Q$ defined by $Q_{i}=R_{i}|_{\left\{ c,x\right\}
}R_{\pi \left( i\right) }|_{B}$ for all $i\in N$. Note that $Q^{\left\{
c,x\right\} }=Q$ and that $R_{i}|_{\left\{ c,x\right\} }=Q_{i}|_{\left\{
c,x\right\} }$ for all $i\in N$, then by Proposition \ref{prop TR TM to =12}%
, $C\left( Q^{\left\{ c,x\right\} }\right) =C\left( R^{\left\{ c,x\right\}
}\right) =$ $\left\{ c,x\right\} $. Thus $C\left( Q\right) =\left\{
c,x\right\} $.\medskip

Now let $H$ be the profile obtained from $Q$ by only permuting $x$ and $c$.
That is $H=Q_{\sigma }$ where $\sigma $ is the permutation of $A$ defined by
$\sigma \left( x\right) =c$, $\sigma \left( x\right) =c$ and $\sigma \left(
z\right) =z$ for all $z\in A\backslash \left\{ c,x\right\} $. Then $%
H_{i}=cxR_{\pi \left( i\right) }|_{B}$ if $i\in N_{1}$, $H_{i}=xcR_{\pi
\left( i\right) }|_{B}$ if $i\in N_{2}$ and $H_{i}=\left( cx\right) R_{\pi
\left( i\right) }|_{B}$ if $i\in N_{3}$ By TS (particularly TN) , $C\left(
H\right) =C\left( \sigma \left( Q\right) \right) =\sigma \left( \left\{
c,x\right\} \right) =\left\{ c,x\right\} $. Also consider the profile $%
U=H_{\pi }$. Since $\pi ^{-1}=\pi $, it follows that: (i) for all $i\in N_{1}
$, $\pi \left( i\right) \in S_{2}$ and $U_{i}=H_{\pi \left( i\right)
}=xcR_{\pi \left[ \pi \left( i\right) \right] }|_{B}=xcR_{i}|_{B}=R_{i}$;
(ii) for all $i\in T_{2}$, $\pi \left( i\right) =i\in N_{2}$ and $%
U_{i}=H_{i}=xcR_{i}|_{B}$; (iii) for all $i\in S_{2}$, $\pi \left( i\right)
\in N_{1}$ and $U_{i}=H_{\pi \left( i\right) }=cxR_{\pi \left[ \pi \left(
i\right) \right] }|_{B}=cxR_{i}|_{B}=R_{i}$; and (iv) for all $i\in N_{3}$, $%
\pi \left( i\right) =i\in N_{3}$ and $U_{i}=H_{\pi \left( i\right) }=\left(
cx\right) R_{\pi \left[ \pi \left( i\right) \right] }|_{B}=\left( cx\right)
R_{i}|_{B}=R_{i}$. By TA, $C\left( U\right) =C\left( H\right) =\left\{
c,x\right\} $.\medskip

Finally, note that $U$ is exactly the profile obtained from $R^{\left\{
c,x\right\} }$ by only reversing the relative ranking of $c$ and $x$ for
each player in $T_{2}$. Since $x\in C\left( R^{\left\{ c,x\right\} }\right) $
and $R\trianglerighteq ^{x,c}U$, then by TM, $C\left( U\right) =\left\{
x\right\} $. This is a contradiction since $C\left( U\right) =\left\{
c,x\right\} $.\medskip

In conclusion, there exists no $x\in C\left( R\right) $ such that $x\notin
CW\left( R\right) $. Thus $C\left( R\right) \subseteq CW\left( R\right) $.

\medskip

\textbf{Necessity:} Assume that an SDR $C$ is Condorcet-consistent. Let
prove that $C$ satisfies TC, TS, TM and TR.

\begin{itemize}
\item Let us prove that $C$ satisfies TC: consider $R\in W^{N}$ and $x\in A$
such that $x\in \mathcal{N}_{C}\left( R\right) $, that is $x\in C\left(
R^{\left\{ x,y\right\} }\right) ,\forall y\in A\backslash \left\{ x\right\} $%
. Therefore by remark \ref{remark CW}, $x\in CW\left( R^{\left\{ x,y\right\}
}\right) ,\forall y\in A\backslash \left\{ x\right\} $ and this means $%
n\left( x,y,R^{\left\{ x,y\right\} }\right) \geq n\left( y,x,R^{\left\{
x,y\right\} }\right) $ $\forall y\in A\backslash \left\{ x\right\} $. This
implies $n\left( x,y,R\right) \geq n\left( y,x,R\right) $ $\forall y\in
A\backslash \left\{ x\right\} $ and therefore $x\in CW\left( R\right) $.
Then $CW\left( R\right) \neq \emptyset $ and since $C$ is
Condorcet-consistent, we have $CW\left( R\right) =C\left( R\right) $ and
therefore $x\in C\left( R\right) $.\medskip

We now prove that $C\left( R\right) =\mathcal{N}_{C}\left( R\right) $. It is
clear with what we just proved that $\mathcal{N}_{C}\left( R\right)
\subseteq C\left( R\right) $. Consider $a\in C\left( R\right) $, then $a\in
CW\left( R\right) $ and this means $n\left( a,b,R\right) \geq n\left(
b,a,R\right) $ $\forall b\in A\backslash \left\{ a\right\} $. It follows
that $a\in CW\left( R^{\left\{ a,b\right\} }\right) =C\left( R^{\left\{
a,b\right\} }\right) $ $\forall b\in A\backslash \left\{ a\right\} $. This
is $a\in \mathcal{N}_{C}\left( R\right) $.

\item Let us prove that $C$ satisfies TN: consider $R\in W^{N}$, $\sigma $ a
permutation on $A$ and $x,y\in A$:%
\begin{equation*}
\begin{array}{ll}
C\left( \sigma \left( R^{\left\{ x,y\right\} }\right) \right) & =C\left(
\sigma \left( R\right) ^{\left\{ \sigma \left( x\right) ,\sigma \left(
y\right) \right\} }\right) \\
& =CW\left( \sigma \left( R\right) ^{\left\{ \sigma \left( x\right) ,\sigma
\left( y\right) \right\} }\right) \\
& =\sigma \left[ CW\left( R^{\left\{ x,y\right\} }\right) \right] \\
& =\sigma \left[ C\left( R^{\left\{ x,y\right\} }\right) \right]%
\end{array}%
\end{equation*}

\item Let us prove that $C$ satisfies TA: consider $R\in W^{N}$, $\pi $ a
permutation on $N$ and $x,y\in A$:%
\begin{equation*}
\begin{array}{ll}
C\left( R_{\pi }^{\left\{ x,y\right\} }\right) & =C\left( R_{\pi \left(
N\right) }^{\left\{ x,y\right\} }\right) \\
& =CW\left( R_{\pi \left( N\right) }^{\left\{ x,y\right\} }\right) \\
& =CW\left( R^{\left\{ x,y\right\} }\right) \\
& =C\left( R^{\left\{ x,y\right\} }\right)%
\end{array}%
\end{equation*}

$C$ satisfies TA and TN, then $C$ satisfies TS.

\item Let us prove that $C$ satisfies TM: consider $R,Q\in W^{N},\left\{
x,y\right\} \subseteq A$ and $z\in A\backslash \left\{ x\right\} $ such that
$x\in C\left( R^{\left\{ x,y\right\} }\right) $ and $R\vartriangleright
^{x,z}Q$. $x\in C\left( R^{\left\{ x,y\right\} }\right) $ implies that $x\in
CW\left( R^{\left\{ x,y\right\} }\right) $ and therefore $x\in CW\left(
Q^{\left\{ x,y\right\} }\right) $ since $R\vartriangleright ^{x,z}Q$. Then $%
x\in C\left( Q^{\left\{ x,y\right\} }\right) $ since $C$ is
Condorcet-consistent. Let now prove that $z\notin C\left( Q^{\left\{
x,y\right\} }\right) $.

(i) If $z\neq y$, then by Proposition \ref{prop TR TM to xy}, $z\notin
C\left( Q^{\left\{ x,y\right\} }\right) $;

(ii) If $z=y$, then $x\in C\left( R^{\left\{ x,y\right\} }\right) $ and $%
R\vartriangleright ^{x,y}Q$ implies that $CW\left( R^{\left\{ x,y\right\}
}\right) =\left\{ x\right\} $. Hence $y\notin CW\left( R^{\left\{
x,y\right\} }\right) $ and then $y\notin C\left( R^{\left\{ x,y\right\}
}\right) $.
\end{itemize}
\end{proof}

\section{Independence of the axioms}

Theorem \ref{Theorem result} is a characterization of Condorcet-consistent
SDRs by means of four axioms. One may wonder whether these axioms are
minimal or not. As shown below, none of them can dropped.

\paragraph{TC can not be dropped:}

Define the SDR $C_{1}$ as follows:%
\begin{equation*}
\forall R\in W^{N},C_{1}\left( R\right) =\left\{
\begin{tabular}{ll}
$CW\left( R\right) $ & if $R\in W_{2}^{N}$ \\
&  \\
$A$ & otherwise%
\end{tabular}%
\right.
\end{equation*}

In one hand, it can be easily checked that $C_{1}$ satisfies TA,\ TN, TM and
TR but fails to be Condorcet-consistent. In the other hand, $C_{1}$ also
satisfies WTC. Then Theorem \ref{Theorem result} can not be restated by
substituting WTC to TC.

\paragraph{TS can not be dropped:}

As an SDR satisfies TS if it satisfies TA and TN, we therefore prove that
none of these two latter axioms can be dropped.

\subparagraph{TA can not be dropped:}

Given a profile $R$, $X$ a subset of $A$ and $i\in N=\left\{
1,2,...,n\right\} $, we set $top\left( R_{i}|_{X}\right) =\left\{ x\in
X/x\succeq _{R_{i}}y,\forall y\in X\right\} $ the set of all voter $i$'s
most preferred alternatives in $X$. Let%
\begin{equation*}
B_{1}\left( R\right) =top\left( R^{1}\right) \text{ and }\forall i\in
N\backslash \left\{ 1\right\} ,B_{i}\left( R\right) =top\left(
R_{i}|_{B_{i-1}\left( R\right) }\right) \text{.}
\end{equation*}%
Define the SDR $C_{2}$ as follows:%
\begin{equation*}
C_{2}\left( R\right) =B_{n}\left( R\right)
\end{equation*}

Note than $C_{2}$ can be viewed as a serial dictatorship for which voter $1$
first selects the set $B_{1}\left( R\right) $ of his/her best alternatives,
voter $2$ selects the set $B_{2}\left( R\right) $ of his/her best
alternatives from $B_{1}\left( R\right) $ and so on. It can be easily
checked that $C_{2}$ satisfies TC, TN, TM and TR but fails to be
Condorcet-consistent.

\subparagraph{TN can not be dropped:}

Consider $a\in A$ and define the SDR $C_{3}$ as follows:%
\begin{equation*}
\forall R\in W^{N},C_{3}\left( R\right) =\left\{
\begin{tabular}{ll}
$CW\left( R|_{A\backslash \left\{ a\right\} }\right) $ & if $CW\left(
R|_{A\backslash \left\{ a\right\} }\right) \neq \emptyset $ \\
&  \\
$A\backslash \left\{ a\right\} $ & otherwise%
\end{tabular}%
\right.
\end{equation*}

It can be easily checked that $C_{3}$ satisfies TC, TA, TM and TR but fails
to be Condorcet-consistent.

\paragraph{TM can not be dropped}

Define the SDR $C_{4}$ as follows:%
\begin{equation*}
\forall R\in W^{N},C_{4}\left( R\right) =\left\{ x\in A/\forall y\in
A,\exists i\in N,x\succ _{R_{i}}y\right\}
\end{equation*}

Note that $C_{4}\left( R\right) $ is the wellknown set of all Pareto optimal
alternatives in $R$. It can be easily checked that $C_{4}$ satisfies TC, TA,
TN and TR but fails to be Condorcet-consistent.

\paragraph{TR can not be dropped}

Define the SDR $C_{5}$ as follows:%
\begin{equation*}
\forall R\in W^{N},C_{5}\left( R\right) =\left\{
\begin{tabular}{ll}
$A\backslash \left\{ x,y\right\} $ & if $R=R^{\left\{ x,y\right\} }$ \\
&  \\
$A$ & otherwise%
\end{tabular}%
\right.
\end{equation*}

Note that $C_{5}\left( R\right) $ is the set of all alternatives ranked
first or second by at least one voter. It can be easily checked that $C_{5}$
satisfies TC, TA, TN and TM but fails to be Condorcet-consistent.

\section{Condorcet domain and Maskin monotonicity}

One main feature of the characterization result (Theorem \ref{Theorem result}%
) is the use of the top consistency axiom. This axiom, while compelling, may
be viewed as quite strong. As we now show, we also obtain an alternative
characterization of Condorcet-consistent rules weakening $TC$ to $WTC$ and
using the well-known Maskin monotonicity condition. Note that this
alternative characterization only applies to the Condorcet domain, i.e.
preference profiles which always admit a Condorcet winner.

\begin{definition}[\protect\cite{dasgupta1979implementation}, \protect\cite%
{maskin1999nash}]
An SDR $C$ is monotononic (MM) provided that $\forall x\in A$, $\forall
R,Q\in W^{N}$, if ($i$) $x\in C(R)$ and ($ii$) $\forall i\in N$, $\forall
y\in A$, $x\succ _{R_{i}}y\Longrightarrow x\succ _{Q_{i}}y$ and $x\sim
_{R_{i}}y\Longrightarrow x\succeq _{Q_{i}}y$, then $x\in C(Q)$.
\end{definition}

We let $W_{\ast }^{N}$ denote the set of profiles that admit a Condorcet
Winner. That is: $R\in W_{\ast }^{N}\Longleftrightarrow CW(R)\neq \emptyset $%
.

\begin{proposition}
\label{[MM2]}If an $SDR$ $C:W_{\ast }^{N}\rightrightarrows A$ is
Condorcet-consistent, then it satisfies $MM$.
\end{proposition}

\begin{proof}
Take any profile $R$ with $x\in C(R)$ and such that $\mathcal{CW}(R)\neq
\emptyset $. Assume that $C$ is Condorcet-consistent; therefore, it must be
the case that $x$ is a Condorcet winner. In other words, we can write that: $%
n\left( x,y,R\right) \geq n\left( y,x,R\right) ,\forall y\neq x$.

We let $m_{1y}=\#\{i\in N:x\succ _{R_{i}}y\}$, $m_{2y}=\#\{i\in N:x\sim
_{R_{i}}y\}$ and $m_{3y}=\#\{i\in N:y\succ _{R_{i}}x\}$ for each $y\neq x$.
Since $x$ is a $CW$, it follows that $m_{1y}\geq m_{3y}$ for each $y\neq x$.

Consider now any profile $Q$ with $\forall i\in N$, $\forall y\in A$, $%
x\succ _{R_{i}}y\Longrightarrow x\succ _{Q_{i}}y$ and $x\sim
_{R_{i}}y\Longrightarrow x\succeq _{Q_{i}}y$. If we can we prove that $x\in
C(Q)$, we have finished the proof.

Again, we let $p_{1y}=\#\{i\in N:x\succ _{Q_{i}}y\}$, $p_{2y}=\#\{i\in
N:x\sim _{Q_{i}}y\}$ and $p_{3y}=\#\{i\in N:y\succ _{Q_{i}}x\}$ for any $%
y\neq x$.

Since $x\succ _{R_{i}}y\Longrightarrow x\succ _{Q_{i}}y$ , it must be the
case $p_{1y}\geq m_{1y}$. Moreover, we have assumed that $x\sim
_{R_{i}}y\Longrightarrow x\succeq _{Q_{i}}y$ so that no voter with $x\sim
_{R_{i}}y$ is such that $y\sim _{Q_{i}}x$. Finally, the voters with $%
y\succeq _{R_{i}}x$ need not be such that $y\succeq _{Q_{i}}x$. It follows
that $p_{3y}\leq m_{3y}$ for any $y\neq x$.

Combining the previous inequalities, we can write that:
\begin{equation*}
p_{1y}\geq m_{1y}\geq m_{3y}\geq p_{3y}.
\end{equation*}

The previous inequality implies that $x$ is Condorcet Winner in the
preference profile $Q$. Hence since $C$ is Condorcet-consistent, $x\in C(Q)$%
, which finishes the proof.
\end{proof}

This result is surprising since Condorcet-consistent rules do not satisfy
Maskin monotonicity in the unrestricted domain. The next result formalizes
this intuition.

\begin{proposition}
\label{CC fails MM}For $n\geq 3,n\neq 4$ and $m\geq 3$, any
Condorcet-consistent SDR $C:W^{N}\rightrightarrows A$ fails to satisfy MM.
\end{proposition}

\begin{proof}
Assume that $m\geq 3$. Let $C$ be a Condorcet-consistent SDR. Assume that $C$
satisfies MM.

In what follows: (i) $a$, $b$ and $c$ are distinct alternatives and $%
B=A\backslash \left\{ a,b,c\right\} $; (ii) $l$ is a given linear order on $%
B $; and (iii) $xyzl$ corresponds to the linear order in which $x$ is ranked
first, $y$ is second, $z$ is third and alternatives in $B$ are ranked
according to $l$, in this case, note that $\left\{ a,b,c\right\} =\left\{
x,y,z\right\} $.

Consider the profile $R$ define as follow for each case:

\begin{equation*}
\begin{tabular}{|c|c|c|}
\hline
& Number of voters & voter preferences \\ \hline
& $p$ & $abcl$ \\ \cline{2-3}
$n=3p,p\geq 1$ & $p$ & $bcal$ \\ \cline{2-3}
& $p$ & $cabl$ \\ \hline
& $p$ & $abcl$ \\ \cline{2-3}
$n=3p+1,p\geq 2$ & $p$ & $bcal$ \\ \cline{2-3}
& $p$ & $cabl$ \\ \cline{2-3}
& $1$ & $abcl$ \\ \hline
& $p$ & $abcl$ \\ \cline{2-3}
& $p$ & $bcal$ \\ \cline{2-3}
$n=3p+2,p\geq 1$ & $p$ & $cabl$ \\ \cline{2-3}
& $1$ & $abcl$ \\ \cline{2-3}
& $1$ & $cbal$ \\ \hline
\end{tabular}%
\end{equation*}

Assume that $a\in C\left( R\right) $ and consider the profile $Q$ obtained
from $R$ by reversing the relative rankings of $b$ and $c$ in the
preferences of all voters who ranked $a$ at the third position. Then $%
CW\left( Q\right) =\left\{ c\right\} $. Since $C$ is Condorcet-consistent,
then $C\left( Q\right) =\left\{ c\right\} $ and $a\notin C\left( Q\right) $.
But note that this is a contradiction since from $R$ to $Q$, the relative
ranking of $a$ with any other alternatives is preserved and $C$ is MM.

The same raisonning is valid when one assumes that $b\in C\left( R\right) $
or $c\in C\left( R\right) $.
\end{proof}

\begin{proposition}
\label{WTC+MM donne TC}If an SDR satisfies WTC and MM, then it satisfies TC.
\end{proposition}

\begin{proof}
Consider $R\in W^{N}$ such that $\mathcal{N}_{C}(R)\neq \emptyset $. Assume,
by contradiction that $C\left( R\right) \neq \mathcal{N}_{C}(R)$. Since WTC
holds, it must be the case that $\mathcal{N}_{C}(R)\subseteq C(R)$. Take $%
x\in C(R)\backslash \mathcal{N}_{C}(R)$.

It follows that there exists some $y\in A$ such that $x\notin C\left(
R^{\left\{ x,y\right\} }\right) $. Indeed, if there is no such $y$, $x\in
\mathcal{N}_{C}(R)$ since $x\in C\left( R^{\left\{ x,z\right\} }\right) $
for any $z\neq x$.

However, one can check that by construction, $x>_{R_{i}}z$ implies $%
x>_{R_{i}^{\left\{ x,y\right\} }}z$ and $x\sim _{R_{i}}z$ implies $x\succeq
_{R_{i}^{\left\{ x,y\right\} }}z$ for any $i\in N$ and for any $z\neq x$.
Indeed, the profile $R^{\left\{ x,y\right\} }$ is obtained by moving
alternatives $x$ and $y$ to the top without altering their relative ranking
with respect to $R$. Therefore MM implies that $x\in C\left( R^{\left\{
x,y\right\} }\right) $, a contradiction.
\end{proof}

The reader can check that the Proposition \ref{WTC+MM donne TC} holds on the
unrestricted domain while Proposition \ref{[MM2]} is on our restricted
domain. Those previous results then lead to a new result of characterization
since they highlight that in the restricted domain, WTC and MM are
equivalent to TC.

\begin{theorem}
An SDR $C:W_{\ast }^{N}\rightrightarrows 2^{A}$ is Condorcet-consistent if
and only if it satisfies WTC, TS, TM, TR and MM.
\end{theorem}

\begin{proof}
Consider an SDR $C:W_{\ast }^{N}\rightrightarrows 2^{A}$.

\textbf{Sufficiency:} assume that $C$ satisfies WTC, TS, TM, TR and MM.
Therefore, by Proposition \ref{WTC+MM donne TC}, $C$ satisfies TC since it
satisfies WTC and MM. Hence $C$ is Condorcet-consistent by Theorem \ref%
{Theorem result} since $C$ satisfies TC, TS, TM and TR.

\textbf{Necessity:} Assume that $C$ is Condorcet-consistent. In one hand, $C$
satisfies MM by Proposition \ref{[MM2]} and in the other hand, $C$ satisfies
WTC, TS, TM and TR. Hence $C$ satisfies WTC, TS, TM, TR and MM.
\end{proof}

\bigskip
\bibliographystyle{authordate1}
\bibliography{bibliototal.bib}

\begin{thebibliography}{}

\bibitem[\protect\citename{Andjiga {\em et~al.}, }2014]{andjiga2014metric}
Andjiga, Nicolas~G, Mekuko, Aurelien~Y, \& Moyouwou, Issofa. 2014.
\newblock Metric rationalization of social welfare functions.
\newblock {\em Mathematical Social Sciences}, {\bf 72}, 14--23.

\bibitem[\protect\citename{Asan \& Sanver, }2002]{AsanSanver2002}
Asan, G., \& Sanver, R. 2002.
\newblock Another {C}haracterization of the {M}ajority {R}ule.
\newblock {\em Economic Letters}, {\bf 75}, 409--413.

\bibitem[\protect\citename{Campbell \& Kelly, }2000]{CampbellKElly2000}
Campbell, D.E., \& Kelly, J.S. 2000.
\newblock A {S}imple {C}haracterization of {M}ajority {R}ule.
\newblock {\em Economic Theory}, {\bf 15}, 689--700.

\bibitem[\protect\citename{Copeland, }1951]{Copeland}
Copeland, A.H. 1951.
\newblock {\em A Reasonable Social Welfare Function}.
\newblock mimeo, University of Michigan.

\bibitem[\protect\citename{Courtin \& N\'u\~nez, }2014]{CourtinNunez2014}
Courtin, S., \& N\'u\~nez, M. 2014.
\newblock {\em A Map of Approval Voting Equilibria Outcomes.}
\newblock mimeo.

\bibitem[\protect\citename{Cr\'epel \& Rieucau, }2005]{CrepelRieucau2005}
Cr\'epel, P., \& Rieucau, J.N. 2005.
\newblock Condorcet's {S}ocial {M}athematics, a {F}ew {T}ables.
\newblock {\em Social Choice and Welfare}, {\bf 25}, 243--285.

\bibitem[\protect\citename{Dasgupta {\em et~al.},
  }1979]{dasgupta1979implementation}
Dasgupta, Partha, Hammond, Peter, \& Maskin, Eric. 1979.
\newblock The implementation of social choice rules: Some general results on
  incentive compatibility.
\newblock {\em The Review of Economic Studies}, {\bf 46}(2), 185--216.

\bibitem[\protect\citename{Dutta, }1988]{Dutta1988}
Dutta, B. 1988.
\newblock Covering sets and a new Condorcet 
\newblock {\em Journal of Economic Theory}, {\bf 44}, 63--80.

\bibitem[\protect\citename{Elkind \& Slinko, }2012]{Elkind2012}
Elkind, E., Faliszewski~P., \& Slinko, A. 2012.
\newblock Rationalizations of {C}ondorcet-consistent {R}ules via {D}istances of
  {H}amming type.
\newblock {\em Social Choice and Welfare}, {\bf 39}, 891--905.

\bibitem[\protect\citename{Fishburn, }1977]{Fishburn}
Fishburn, P. 1977.
\newblock Condorcet Social Choice Functions.
\newblock {\em SIAM Journal of Applied Mathematics}, {\bf 33}, 469--489.

\bibitem[\protect\citename{Geherlein, }2006]{Gehrlein2006}
Geherlein, W.V. 2006.
\newblock {\em Condorcet's {P}aradox}.
\newblock Springer, Berlin-Heidelberg.

\bibitem[\protect\citename{Henriet, }1985]{Henriet1985}
Henriet, D. 1985.
\newblock The Copeland choice function: an axiomatic characterization.
\newblock {\em Social Choice and Welfare}, {\bf 2}, 49--63.

\bibitem[\protect\citename{Laffond {\em et~al.}, }1993]{Laffond1993}
Laffond, G., Laslier, J.~F., \& Le~Breton, M. 1993.
\newblock The {B}ipartisan {S}et of a {T}ournament {G}ame.
\newblock {\em Games and Economic Behavior}, {\bf 5}, 182--201.

\bibitem[\protect\citename{Laslier, }2009]{Las2}
Laslier, J-F. 2009.
\newblock Strategic {A}pproval {V}oting in a {L}arge {E}lectorate.
\newblock {\em Journal of Theoretical Politics}, {\bf 21}, 113--136.

\bibitem[\protect\citename{Maskin, }1999]{maskin1999nash}
Maskin, Eric. 1999.
\newblock Nash equilibrium and welfare optimality.
\newblock {\em The Review of Economic Studies}, {\bf 66}(1), 23--38.

\bibitem[\protect\citename{May, }1952]{may1952set}
May, Kenneth~O. 1952.
\newblock A set of independent necessary and sufficient conditions for simple
  majority decision.
\newblock {\em Econometrica: Journal of the Econometric Society},  680--684.

\bibitem[\protect\citename{Myerson, }1995]{Myerson1995}
Myerson, R. 1995.
\newblock Axiomatic {D}erivation of {S}coring {R}ules without the {O}rdering
  {A}ssumption.
\newblock {\em Social Choice and Welfare}, {\bf 12}, 59--74.

\bibitem[\protect\citename{P{\'e}rez-Fern{\'a}ndez {\em et~al.},
  }2017]{perez2017monometrics}
P{\'e}rez-Fern{\'a}ndez, Ra{\'u}l, Rademaker, Michael, \& De~Baets, Bernard.
  2017.
\newblock Monometrics and their role in the rationalisation of ranking rules.
\newblock {\em Information Fusion}, {\bf 34}, 16--27.

\bibitem[\protect\citename{Pivato, }2013]{Pivato2013}
Pivato, M. 2013.
\newblock Variable-{P}opulation {V}oting {R}ules.
\newblock {\em Journal of Mathematical Economics}, {\bf 49}, 210--221.

\bibitem[\protect\citename{Schwartz, }1972]{Schwartz1972}
Schwartz, T. 1972.
\newblock Rationality and the {M}yth of the {M}aximum.
\newblock {\em Nous}, {\bf 6}, 97--117.

\bibitem[\protect\citename{Schwartz, }1990]{Schwartz1990}
Schwartz, T. 1990.
\newblock Cyclic {T}ournaments and {C}ooperative {M}ajority {V}oting: A
  {S}olution.
\newblock {\em Social Choice and Welfare}, {\bf 7}, 19--29.

\bibitem[\protect\citename{Slater, }1961]{Slater}
Slater, P. 1961.
\newblock Inconsistencies in a {S}chedule of {P}aired {C}omparisons.
\newblock {\em Biometrika},  303--312.

\bibitem[\protect\citename{Smith, }1973]{Smith73}
Smith, J. 1973.
\newblock Aggregation of {P}references with a {V}ariable {E}lectorate.
\newblock {\em Econometrica}, {\bf 41}, 1027--1041.

\bibitem[\protect\citename{van~der Hout {\em et~al.},
  }2006]{van2006characteristic}
van~der Hout, Eliora, de~Swart, Harrie, \& ter Veer, Annemarie. 2006.
\newblock Characteristic properties of list proportional representation
  systems.
\newblock {\em Social Choice and Welfare}, {\bf 27}(3), 459--475.

\bibitem[\protect\citename{Woeginger, }2003]{Woeginger}
Woeginger, G. 2003.
\newblock A {N}ew {C}haracterization of the {M}ajority {R}ule.
\newblock {\em Economic Letters}, {\bf 81}, 89--94.

\bibitem[\protect\citename{Young, }1975]{Young1975}
Young, P. 1975.
\newblock Social {C}hoice {S}coring {F}unctions.
\newblock {\em SIAM Journal of Applied Mathematics}, {\bf 27}, 824--838.

\end{thebibliography}

\end{document}